\documentclass[a4paper,11pt]{article}
\usepackage[utf8]{inputenc}
\usepackage[T1]{fontenc}
\usepackage{tgtermes}
\usepackage{amsfonts,amsmath,amssymb,amsthm}
\usepackage{enumerate}
\usepackage[linesnumbered,ruled,vlined]{algorithm2e}

\providecommand{\DontPrintSemicolon}{\dontprintsemicolon}

\usepackage{color}

\usepackage{tikz}
\usepackage{url}

\usepackage[hidelinks]{hyperref}

\usepackage[margin=1.2in]{geometry}

\newtheorem{theorem}{Theorem}[section]
\newtheorem{lemma}[theorem]{Lemma}

\newcommand{\indep}{\textsc{Maximum Independent Set}}
\newcommand{\verte}{\textsc{Minimum Vertex Cover}}
\newcommand{\domi}{\textsc{Minimum Dominating Set}}

\newcommand{\OO}{\operatorname{O}}

\renewcommand{\S}{\mathcal{S}}

\title{Simple PTAS's for families of graphs excluding a minor\thanks{This work is partially funded by the Slovenian Research Agency.}}
\author{
Sergio Cabello\thanks{Department of Mathematics, IMFM, and Department of Mathematics, FMF, University of
Ljubljana, Slovenia. \href{mailto:sergio.cabello@fmf.uni-lj.si}{sergio.cabello@fmf.uni-lj.si}}
\and
David Gajser\thanks{Department of Mathematics, IMFM, Slovenia. \href{mailto:david.gajser@fmf.uni-lj.si}{david.gajser@fmf.uni-lj.si}}
}

\begin{document}
	\thispagestyle{empty}
	\maketitle

\begin{abstract}
	We show that very simple algorithms based on local search are polynomial-time approximation schemes for \indep, \verte\ and \domi, when the input graphs
	have a fixed forbidden minor.
\end{abstract}

\section{Introduction}

In this paper we present very simple PTAS's (polynomial-time approximation schemes) based on greedy local optimization for \indep, \verte\ and \domi\ in minor-free families of graphs. The existence of PTAS's for such problems was shown by Grohe~\cite{Grohe03}, and better time bounds were obtained using the framework of bidimensionality; see the survey~\cite{dh-08} and references therein.  The advantage of our algorithms is that they are surprisingly simple and do not rely on deep structural results for minor-free families.

A graph $H$ is a \emph{minor} of $G$ if $H$ can be obtained from a subgraph of $G$ by edge contractions. We say that $G$ is \emph{$H$-minor-free} if $H$ is not its minor. A family of graphs is \emph{$H$-minor-free} if all the graphs in the family are $H$-minor-free.  It is well-known that the family of planar graphs is $K_{3,3}$-minor-free and $K_5$-minor-free, and similar results hold for graphs on surfaces. Thus, minor-free families is a vast extension of the family of planar graphs and, more generally, graphs on surfaces. We will restrict our attention to $K_h$-minor-free graphs, where $K_h$ is the complete graph on $h$ vertices, because $H$-minor-free graphs are also $K_{|V(H)|}$-minor-free.

The development of PTAS's for graphs with a forbidden fixed minor is often based on a complicated theorem of Robertson and Seymour~\cite{rs-03} describing the structure of such graphs. In fact, one needs an algorithmic version of the structural theorem and much work has been done to obtain simpler and faster algorithms finding the decomposition.  See Grohe, Kawarabayashi and Reed~\cite{gkr-13} for the latest improvement and a discussion of previous work.  Even those simplifications are still very complicated and, in fact, the description of the structure of $K_h$-minor-free graphs is cumbersome in itself.  Obtaining a PTAS for  \indep\ restricted to minor-free families is easier and can be done through the computation of separators, as shown by Alon, Seymour, and Thomas~\cite{ast-90}.  However, the approach does not work for \verte\ and \domi.  Baker~\cite{Baker} developed a technique to obtain PTAS for planar graphs using more elementary tools.  In fact, much of the work for minor-free families is a vast, complex generalization of the approach by Baker.

To show how simple is our approach, look at the algorithm \textsc{Independent($h,G,\varepsilon$)} for \indep\ shown in Figure~\ref{fig:algI}. The algorithms for \verte\ and \domi\ are similar and provided in Section~\ref{mainPart}. In the algorithms we use a constant $C_h$ that depends only on the size of the forbidden minor. Its actual value is in $\Theta(h^3)$, as we shall see.

\begin{figure}
	\begin{algorithm}[H]
		\DontPrintSemicolon
		\KwIn{An integer $h>0$, a $K_h$-minor-free graph $G=(V,E)$, and a parameter $\varepsilon\in(0,1)$}
		\KwOut{An independent set $U$ of $G$}
		$r= C_h/\varepsilon^2$, where $C_h$ is an appropriate constant depending on $h$\;
		$U=\emptyset$\;
		\While{$\exists U_1\subseteq U$, $ V_1\subseteq V\setminus U$ with $|U_1|<|V_1|\leq r$ and $(U\backslash U_1)\cup V_1$ is an independent set}{
			$U=(U\backslash U_1)\cup V_1$\;
		}
		\KwRet{$U$}
		\caption{\textsc{Independent($h,G,\varepsilon$)}}
	\end{algorithm}
	\caption{PTAS for \indep\ for $K_h$-minor-free graphs.}
	 \label{fig:algI}
\end{figure}

We see that, for any fixed $h$, the algorithm is a very simple local optimization that returns an independent set that is $\OO(\varepsilon^{-2})$-locally optimal, in the sense that it cannot be made larger by substituting any $\OO(\varepsilon^{-2})$ of its vertices.  The algorithm runs in time $n^{\OO(\varepsilon^{-2})}$, for any fixed $h$.

The main idea in the proof of the correctness of our algorithm is dividing the input graph into not-too-many pieces with $\OO(\varepsilon^{-2})$ vertices and small boundary, as defined in Section~\ref{division}.  For this we use the existence of separators~\cite{ast-90} in the same way as Frederickson~\cite{Frederickson} did for planar graphs.  A similar division has been used in other works; see for example~\cite{Yuster}.  The division is useful for the following fact: changing the solution $U$ within one of the pieces can not result in a better solution because $U$ is $\OO(\varepsilon^{-2})$-locally optimal. Using this, we can infer (after some work) that, if $G$ is $K_h$-minor-free, then $$opt-|U|\leq\varepsilon \cdot opt.$$ For \verte\ and \domi\ one has to make the additional twist of considering a division in a graph that represents the locally optimal solution and the optimal solution.

It is important to note that the analysis of the algorithm uses separators but the algorithm does not use them.  Thus, all the difficulty is in the proof that the algorithm is a PTAS, not in the description of the algorithm.  In any case, our proofs only rely on the existence of separators and is dramatically simpler than previous proofs of existence of PTAS's for \verte\ and \domi. In particular, we do not need any of the tools developed for the Graph Minor Theorem. A drawback of our method is that the running time is $n^{\OO(\varepsilon^{-2})}$, while previous, more complicated methods require $\OO( f(\varepsilon) n^c)$, for some constant $c>0$ and function $f$. Another drawback of our method is that it works only in unweighted problems.

The idea of using separators to show that a local-optimization algorithm is a PTAS was presented by Chan and Har-Peled~\cite{Chan} and independently by Mustafa and Ray~\cite{mr10}.  Local search was also used earlier to obtain constant-factor approximations by Agarwal and Mustafa~\cite{am-06}.  The technique has been used recently to provide PTAS's for some geometric problems; see for example~\cite{bgmr-15,cm-14,gp-10,kgkv-14}.  However, the use for minor-free families of graphs has passed unnoticed.
\section{Dividing minor-free graphs}
	\label{division}

In this section we present a way of dividing a graph into subgraphs 
with special properties. We will not use this division in our algorithms, 
but it will be the main tool for their analysis.
	
Let $G$ be a graph and let $\S = \{ S_1, \dots, S_k\}$ be a collection 
of subsets of vertices of $G$.
We define the \emph{boundary} of a piece $S_i\in \S$ (with respect to $\S$), 
denoted by $\partial S_i$, as those vertices of $S_i$ that appear in some
other piece $S_j\in \S$, $j\not= i$.
Thus $\partial S_i = S_i \cap \left( \bigcup_{j\not= i} S_j \right)$.  
We define the \emph{interior} of $S$ as $\mathrm{int}(S_i)= S_i\setminus\partial S_i$. 

A \emph{division} of a graph $G$ is a collection $\S = \{ S_1, S_2, \dots , S_k\}$
of subsets of vertices of $G$ satisfying the following two properties:
\begin{itemize}
	\item $G=\bigcup_i G[S_i]$, that is, 
		each edge and each vertex of $G$ appears in some induced subgraph $G[S_i]$, and
	\item for each $S_i\in \S$ and $v\in \mathrm{int}(S_i)$, all neighbours of $v$ are in $S_i$.
\end{itemize}
We refer to each subset $S_i\in \S$ as a \emph{piece} of the division.
(It may be useful to visualize a piece as the induced subgraph $G[S_i]$,
since we actually use $S_i$ as a proxy to $G[S_i]$.)

We want to find a division of a $K_h$-minor-free graph $G$ where, 
for some parameter $r$ that we can choose, each piece has roughly $r$ vertices
and all pieces together have roughly $|V(G)|/\sqrt{r}$
boundary vertices, counted with multiplicity.
For technical reasons explained below, we will consider only the case when $r\ge \Omega(h^3)$.
We will prove the following, without trying to optimize the constants involved.

\begin{lemma}
	\label{divisible}
	For each $K_h$-free-minor graph $G$ with $n$ vertices 
	and any $r$ with $36h^3 \le r\le n$, 
	there exists a division $\{ S_1, \dots ,S_k\}$ of $G$ satisfying
	the following two properties:
	\begin{itemize}
		\item $|S_i|\leq r$ for $i=1,\dots,k$, and
		\item $\sum_i |\partial S_i|\leq \frac{36 h ^{3/2} n}{\sqrt{r}}$.
	\end{itemize}
\end{lemma}

For planar graphs, a stronger lemma was proven by Frederickson~\cite{Frederickson}. 
We refer to the draft by Klein and Mozes~\cite{km-2014} for a careful treatment.
The proof for the more general $K_h$-minor-free case is very similar. 
However, we have not been able to find a careful treatment for $K_h$-minor-free graphs
and thus decided to include a proof where the dependency on $h$ is explicit.
The main tool in the proof is the separator theorem for $K_h$-minor-free graphs
proven by Alon, Seymour and Thomas~\cite{ast-90}. 
It states that in every $K_h$-minor-free graph $G$ with $n$ vertices, 
there exists a partition of vertices of $G$ into three sets $A$, $B$ and $X$ 
such that $|A|,|B|\le \tfrac{2}{3}n$, $|X|\leq h^{3/2}\sqrt{n}$ and no edge connects a vertex from $A$ to a vertex from $B$. The set $X$ is called the \emph{separator}.

We are going to use the separator theorem recursively.
We restrict our attention to the case $r\ge \Omega(h^3)$ because, when we
get to subgraphs with $h^3$ vertices,
the size of the separator guaranteed by the separator theorem is also $h^3$,
and thus we cannot benefit from recursion anymore.

\begin{proof}[Proof of Lemma~\ref{divisible}.]
	Let $G$ be a $K_h$-minor free graph with $n$ vertices
	and assume that $r \ge 36 h^3$. 

	Consider the following 
	algorithm to compute a division into pieces of size $r$. 
	We start setting $\mathcal{S}= \{ V(G) \}$.
	While $\mathcal{S}$ has some piece $S$ with more than $r$ vertices,
	we remove $S$ from $\mathcal{S}$, 
	use the separator theorem on the induced subgraph $G[S]$ to obtain sets $A$, $B$ and $X$,
	and put the pieces $A\cup X$ and $B\cup X$ in $\mathcal{S}$.
	This finishes the description of the construction.

	Whenever we apply the separator theorem to a piece $S$ with more than $r\ge 36 h^3$
	vertices, the sets $A\cup X$ and $B\cup X$ are strictly smaller than $S$. 
	Thus, the algorithm finishes.
	Since in each iteration of the construction the separator goes to both subpieces,
	we maintain a division. Formally, one could show by induction on the number
	of iterations that $\S$ is always a division.
	By construction, each of the pieces in the resulting division has
	at most $r$ vertices. It remains to bound the sum of the size of the boundaries.
	
	Let $S$ be any of the pieces considered through the algorithm,
	and assume that the construction subdivides $S$ into pieces $S_1,\dots, S_t$.
	We define 
	\[
		\beta(S) := \left(\sum_{i=1}^t |S_i| \right) - |S|.
	\]
	Thus $\beta(S)$ is the sum of the sizes of the final pieces
	obtained through the recursive partitioning of $S$, minus the size of $S$.	
	Let $\beta(m):= \max \beta(S)$, where the maximum is taken over all
	pieces $S$ with $m$ vertices that appear through the construction.
	We want to bound $\beta(V(G)) = \beta(n)$.
	
	When we break a piece $S$ with $m$ vertices into two pieces $S_1$ and $S_2$
	using a separator $X$ of size $h^{3/2}\sqrt{m}$, we have
	$|S_1| + |S_2| \le |S| + h^{3/2}\sqrt{m}$ and therefore
	$\beta(S) \le \beta(S_1) + \beta(S_2) + h^{3/2}\sqrt{m}$.
	Thus, for every $m$, there exist $m_1$ and $m_2$ such that we have the recurrence
	\begin{equation}
		\beta(m) ~\le~ 
		\begin{cases}
			\beta (m_1)+ \beta (m_2)+ h^{3/2}\sqrt{m}&
					\text{ if $m> r$,}\\
			0 &\text{ if $m\leq r$,}
		\end{cases} \label{eq1}
	\end{equation}
	where $m_1,m_2\ge m/3$ and $m_1+m_2 \le m+h^{3/2}\sqrt{m}$.
	It follows by induction that 
	\begin{equation*}
		\beta (m) ~\leq~ 
		\begin{cases}
			\frac{(10 h^{3/2}) m}{\sqrt{r/3}}  -  (10 h^{3/2}) \sqrt{m},   & \mbox{if } m\geq\frac{r}{3} \\
			0&				\mbox{otherwise.}
		\end{cases} 
	\end{equation*}
	The proof is a standard computation and we include it in Appendix~\ref{app}
	for completeness.
	
	Consider the division $\S=\{ S_1,\dots, S_k\}$ constructed by the algorithm,
	and define $\partial = \bigcup_{i} \partial S_i$. Thus $\partial$ is the set of all vertices that
	are boundary of some piece $S_i\in \S$. Since each vertex in $\partial$ is boundary
	in at least 2 pieces, we have $2\cdot |\partial| \le \sum_i |\partial S_i|$ 
	and therefore
	\begin{align*}
		\sum_i |\partial S_i| ~&\le~ 2 \left( \sum_i |\partial S_i|-|\partial|\right) ~=~
		2 \left( \sum_i  |S_i|-n \right) ~=~ 2\cdot \beta(V(G)) \\
		&\le~ 2\cdot \beta(n) ~\le~ 2 \frac{(10 h^{3/2}) n}{\sqrt{r/3}} ~<~
		\frac{36 h^{3/2} n}{\sqrt{r}}.
	\end{align*}
	We conclude that the algorithm has constructed a division
	with the desired properties.
\end{proof}
\section{Algorithms and analyses}
	\label{mainPart}
In this section we present and analyze PTAS's for the problems \indep, \verte\ and \domi\ restricted on $K_h$-minor-free graphs. All algorithms will be simple local optimizations, as discussed in the introduction.

\subsection{Independent set}
	\label{mainIndep}
Consider the algorithm \textsc{Independent($h,G,\varepsilon$)} for \indep\ that was given 
in Figure~\ref{fig:algI}.
We will show that, for any fixed constant $h$, the algorithm is a PTAS.

\begin{theorem}
	\label{ind}
	For any fixed integer $h>0$, the algorithm \textsc{Independent($h,G,\varepsilon$)} 
	is a PTAS for the problem \indep\ restricted to $K_h$-minor-free graphs
	with running time $n^{\OO(1/\varepsilon^2)}$. 
	The constant hidden in the big $\OO$ in the running time is polynomial in $h$.
\end{theorem}
\begin{proof}
   	Set $C_h= 144^2 h^3$ in the algorithm \textsc{Independent($h,G,\varepsilon$)}.
	Let $G$ be a $K_h$-minor-free graph, 
   let $U^*$ be a largest independent set of $G$, and let $U$ be the independent set returned by
   the algorithm. We have to show that  
	\begin{equation*}
		\label{enacbaI}
		|U^*|-|U| ~\le~ \varepsilon |U^*|.
	\end{equation*}
	Consider the induced subgraph $\tilde G=G[U\cup U^*]$, which is also $K_h$-minor-free.
	Let $\{ S_1,\dots ,S_k\}$ be the division of $\tilde G$ guaranteed 
	by Lemma~\ref{divisible} for $r= C_h/\varepsilon^2$, 
	as set in \textsc{Independent($h,G,\varepsilon$)}.
	Note that $r\ge 36h^3$ satisfies the requirements for Lemma~\ref{divisible}.
	
	Note that a subset of $U\cup U^*$ is independent in $G$ if and only if
	it is independent in $\tilde G$. Therefore, 
	\[
		\forall i\in [k]: ~~~ (U\backslash  S_i) \cup (U^*\cap \mathrm{int}(S_i)) \text{ is an independent set in $G$}.
	\]
	By the algorithm, the independent set $U$ can not be made larger by any
	such a transformation, thus we have
	\[
		\forall i\in [k]: ~~~ |U|~\geq~ |U|-|U\cap S_i| + |U^*\cap \mathrm{int}(S_i)|,
	\]
	or alternatively
	\[
		\forall i\in [k]: ~~~ |U\cap S_i| ~\geq~ |U^*\cap \mathrm{int}(S_i)|.
	\]
	We can use this inequality, summed over all $i\in [k]$, to get
	\begin{align*}
		|U^*| ~&\leq~ \sum_i |U^*\cap \mathrm{int}(S_i)| + \sum_i|\partial S_i|\\
				&\leq~ \sum_i |U\cap S_i| + \sum_i |\partial S_i|\\
				&\leq~ |U| + 2\sum_i |\partial S_i|.
	\end{align*}
	Using the bound $\sum_i |\partial S_i|\le \frac{36 h^{3/2} \cdot |U\cup U^*|}{\sqrt{r}}$ 
	from Lemma~\ref{divisible} and substituting $r$, we get 
	\begin{align*}
		|U^*|-|U|  ~\leq~ 2\frac{36 h^{3/2}\cdot |U\cup U^*|}{\sqrt{r}}
					~\le~ \frac{72 h^{3/2}\cdot 2 \cdot |U^*|}{\sqrt{144^2 h^3/\varepsilon^2}}
					~=~ \varepsilon |U^*|.
	\end{align*}
	The running time is $n^{\OO(r)}=n^{\OO(C_h/\varepsilon^2)}= n^{\OO(h^3/\varepsilon^2)}$.
\end{proof}

\subsection{Vertex cover}
	\label{mainVert}
For \verte\ consider the greedy local optimization algorithm \textsc{Vertex($h,G,\varepsilon$)}
given in Figure~\ref{fig:algV}. Its structure is very similar to the algorithm for \indep.

\begin{figure}
	\begin{algorithm}[H]
		\DontPrintSemicolon
		\KwIn{An integer $h>0$, a $K_h$-minor-free graph $G=(V,E)$ and a 
			parameter $\varepsilon\in(0,1)$}
		\KwOut{A vertex cover $U$ of $G$}
		$r=C_h/\varepsilon^2$, where $C_h$ is an appropriate constant depending on $h$\;
		$U=V$\;
		\While{$\exists U_1\subseteq U$, $ V_1\subseteq V\setminus U$ with $|V_1|<|U_1|\leq r$ and $(U\backslash U_1)\cup V_1$ is a vertex cover}{
			$U=(U\backslash U_1)\cup V_1$\;
		}
		\KwRet{$U$}
		\caption{\textsc{Vertex($h,G,\varepsilon$)}}
	\end{algorithm}
	\caption{PTAS for \verte\ for $K_h$-minor-free graphs.}
	\label{fig:algV}
\end{figure}

\begin{theorem}
	\label{vert}
	For any fixed integer $h>0$, the algorithm \textsc{Vertex($h,G,\varepsilon$)} 
	is a PTAS for the problem \verte\ restricted to $K_h$-minor free graphs 
	with running time $n^{\OO(1/\varepsilon^2)}$. 
	The constant hidden in the big $\OO$ in the running time is polynomial in $h$.
\end{theorem}
\begin{proof}
	The proof is very similar to the proof for \indep. We do not attempt to shorten it
	and follow very much the same structure.
	
	Set $C_h= 4\cdot 144^2 h^3$ in the algorithm \textsc{Vertex($h,G,\varepsilon$)}.
	Let $G$ be a $K_h$-minor-free graph, 
    let $U^*$ be a smallest vertex cover of $G$, 
	and let $U$ be the vertex cover returned by the algorithm. We have to show that  
	\begin{equation*}
		\label{enacbaV}
		|U|- |U^*| ~\leq~\varepsilon |U^*|.
	\end{equation*}
	Consider the induced subgraph $\tilde G= G[U\cup U^*]$, which is also $K_h$-free-minor.
	Let $\{ S_1, \dots ,S_k\}$ be the division of $\tilde G$ guaranteed 
	by Lemma~\ref{divisible} for $r= C_h/\varepsilon^2$, 
	as set in \textsc{Vertex($h,G,\varepsilon$)}.
	Note that $r\ge 36h^3$ satisfies the requirements for Lemma~\ref{divisible}.	

	Consider an edge $uv$ of $G$ and assume that $u\in U\cap \mathrm{int}(S_i)$.
	If $u\notin U^*$, then $v\in U^*$ and $uv\in E(\tilde G)$, 
	which implies that $v\in U^*\cap S_i$ because $u\in \mathrm{int}(S_i)$.
	We conclude that $u$ or $v$ are in $U^*\cap S_i$.
	Therefore
	\[
		\forall i\in [k]: ~~~ (U\backslash \mathrm{int}(S_i)) \cup (U^*\cap  S_i) \text{ is a vertex cover}.
	\]
	By the algorithm, the vertex set $U$ can not be made smaller by any
	such a transformation, thus we have
	\[
		\forall i\in [k]: ~~~ |U|\leq |U|-|U\cap \mathrm{int}(S_i)| + |U^*\cap S_i|,
	\]
	or alternatively
	\[
		\forall i\in [k]: ~~~ |U\cap \mathrm{int}(S_i)| ~\leq~ |U^*\cap S_i|.
	\]
	We can use this inequality, summed over all $i\in [k]$, to get
	\begin{align*}
		|U^*| ~& \geq~ \sum_i |U^*\cap S_i|-\sum_i |\partial S_i|\\
			   &\geq~ \sum_i |U\cap \mathrm{int}(S_i)|-\sum_i|\partial S_i|\\
				&\geq~ |U|-2\sum_i |\partial S_i|.
	\end{align*}
	Using the bound $\sum_i |\partial S_i|\le \frac{36 h^{3/2} \cdot |U\cup U^*|}{\sqrt{r}}$ 
	from Lemma~\ref{division} and substituting $r$, we get 
	\begin{align*}
		|U|-|U^*|  ~\leq~ 2\frac{36 h^{3/2} \cdot |U\cup U^*|}{\sqrt{r}}
					~\le~ \frac{72 \cdot 2 \cdot |U|}{\sqrt{4\cdot 144^2/\varepsilon^2}}
					~=~ \frac{\varepsilon}{2} |U|,
	\end{align*}
	which implies 
	\[
		|U| ~\le ~ \frac{1}{1-\varepsilon/2} \cdot |U^*| ~\leq~ (1+\varepsilon)\cdot |U^*|
	\]
	for $\varepsilon\in (0,1)$.
	The running time is $n^{\OO(r)}=n^{\OO(C_h/\varepsilon^2)}= n^{\OO(h^3/\varepsilon^2)}$.
\end{proof}

\subsection{Dominating set}
	\label{mainDomi}
The PTAS for the problem \domi\ on $K_h$-minor-free families of graphs 
is practically the same as the algorithm \textsc{Vertex($h,G,\varepsilon$)}. 
We call it \textsc{Dominating($h,G,\varepsilon$)} and
include it in Figure~\ref{fig:algD} to reference to it.

\begin{figure}
	\begin{algorithm}[H]
		\DontPrintSemicolon
		\KwIn{An integer $h>0$, a $K_h$-minor-free graph $G=(V,E)$ and a 
			parameter $\varepsilon\in(0,1)$}
		\KwOut{A dominating set $U$ of $G$}
		$r=C_h/\varepsilon^2$, where $C_h$ is an appropriate constant depending on $h$\;
		$U=V$\;
		\While{$\exists U_1\subseteq U$, $ V_1\subseteq V\setminus U$ with $|V_1|<|U_1|\leq r$ and $(U\backslash U_1)\cup V_1$ is a dominating set}{
			$U=(U\backslash U_1)\cup V_1$\;
		}
		\KwRet{$U$}
		\caption{\textsc{Dominating($h,G,\varepsilon$)}}
	\end{algorithm}
	\caption{PTAS for \domi\ for $K_h$-minor-free graphs.}
	\label{fig:algD}
\end{figure}

\begin{theorem}
	\label{domi}
	For any fixed integer $h>0$, the algorithm \textsc{Dominating($h,G,\varepsilon$)} 
	is a PTAS for the problem \domi\ restricted to $K_h$-minor-free graphs
	with running time $n^{\OO(1/\varepsilon^2)}$.
	The constant hidden in the big $\OO$ in the running time is polynomial in $h$.
\end{theorem}

Although the algorithms \textsc{Vertex($h,G,\varepsilon$)} and \textsc{Dominating($h,G,\varepsilon$)} are almost identical, we need an additional idea in the analysis of the latter.
\begin{proof}
	Set $C_h= 4\cdot 144^2 h^3$ in algorithm \textsc{Dominating($h,G,\varepsilon$)}.
	Let $G$ be a $K_h$-minor-free graph, 
    let $U^*$ be a smallest dominating set of $G$, 
	and let $U$ be the dominating set returned by the algorithm. We have to show that  
	\begin{equation*}
		\label{enacbaD}
		|U|- |U^*| ~\leq~\varepsilon |U^*|.
	\end{equation*}
	
	If we would take a division $\{S_1,\dots , S_k\}$ of the induced graph $G[U\cup U^*]$, 
	as in the vertex cover case, then the sets
	\[
		(U\backslash \mathrm{int}(S_i)) \cup (U^*\cap  S_i)
	\]
	would not necessarily be dominating; see Figure~\ref{example}.
	This is crucial for the argument to go through, so we need to proceed differently.

	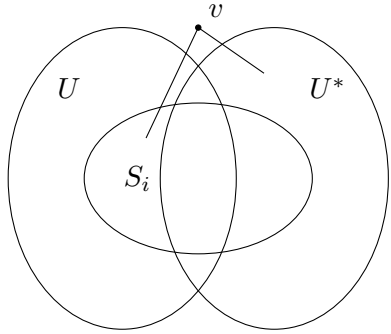
\begin{figure}[!htb]
		\begin{center}
		\begin{tikzpicture}
			\draw (-1,0) ellipse (1.5cm and 2cm);
			\draw (1,0) ellipse (1.5cm and 2cm);
			\draw (0,0) ellipse (1.5cm and 1cm);
				\node at (-1.7,1.2) (U) {$U$};
				\node at (1.7,1.2) (U1) {$U^*$};
				\node at (-0.8,0) (S) {$S_i$};
			\node at (-0.75,0.4) (u) {};
			\node at (1,1.3) (u1) {};
			\node[above right] at (0,2) {$v$};
			\draw[fill] (0,2) circle [radius=1pt]
				edge [-] (u)
				edge [-] (u1);
		\end{tikzpicture}
		\end{center}
		\caption{If $\{S_1,\dots ,S_k\}$ is a division of the graph $G[U\cup U^*]$, 
			then the vertex $v$ might be dominated by $U$ but not by 
			$(U\backslash \mathrm{int}(S_i)) \cup (U^*\cap  S_i)$.}
		\label{example}
	\end{figure}

	For every vertex $v\in V\backslash (U\cup U^*)$, choose an edge that connects this vertex to $U$ 
	and contract it. Such an edge exists because $U$ is a dominating set. 
	Let $\tilde{G}$ be the resulting graph. Its vertex set is $U\cup U^*$. 
	It is clear that $\tilde{G}$ is $K_h$-minor-free, since it is a minor of $G$.
	Let $\{ S_1, \dots ,S_k\}$ be the division of $\tilde G$ guaranteed 
	by Lemma~\ref{divisible} for $r= C_h/\varepsilon^2$.
 
	We claim that, for each index $i$, the set
	$$U_i=(U\backslash \mathrm{int}(S_i)) \cup (U^*\cap  S_i)$$
	is dominating in $G$.
	We do not know of a better way to verify this than by a systematic case-by-case analysis 
	over all vertices $v\in V$.
	\begin{description}
	\item[Case $\mathbf{v\in V\backslash(U\cup U^*).}$] Because $U$ and $U^*$ are dominating sets, there exist vertices $u\in U$ and $u^*\in U^*$ that are neighbours of $v$. Without loss of generality we may assume that the edge $uv$ was contracted when $\tilde{G}$ was constructed. If $u\not\in U_i$, it must be that $u\in\mathrm{int}(S_i)$, which implies that $u^*\in S_i$, since $u^*$ is a neighbor of $u$ in $\tilde{G}$. Hence, $u\in U_i$ or $u^*\in U_i$, which implies that $v$ is dominated by $U_i$.
	\item[Case $\mathbf{v\in U.}$] Because $U^*$ is a dominating set, there exist a vertex $u^*\in U^*$ that dominates $v$ in $G$. If $v\not\in U_i$, then $v\in\mathrm{int}(S_i)$ which implies $u^*\in S_i$, thus $u^*\in U_i$. Hence, $v$ is dominated by $U_i$.
	\item[Case $\mathbf{v\in U^*.}$] Because $U$ is a dominating set, there exist a vertex $u\in U$ that dominates $v$ in $G$. If $u\not\in U_i$, then $u\in\mathrm{int}(S_i)$ which implies $v\in S_i$, thus $v\in U_i$. Hence, $v$ is dominated by $U_i$.
	\end{description}
	Our intuition for why dividing $\tilde{G}$ is better than dividing $G[U\cup U^*]$ 
	is that $\tilde{G}$ has all the edges of $G[U\cup U^*]$ plus some additional ones 
	and hence the division of $\tilde{G}$ is stronger.

	The proof from here on is identical to the one in the vertex cover case.
\end{proof}

\bibliographystyle{abuser}
\bibliography{literature}

\begin{appendix}
\section{Computation for Lemma~\ref{divisible}}
\label{app}
	We want to show by induction that
	\begin{equation*}
		\beta (m) ~\leq~ \beta_{ind}(m) ~:=~
		\begin{cases}
			\frac{(10 h^{3/2}) m}{\sqrt{r/3}}  -  (10 h^{3/2}) \sqrt{m},   & \mbox{if } m\geq\frac{r}{3} \\
			0&				\mbox{otherwise.}
		\end{cases} \label{eq2}
	\end{equation*}

	We first check the base case. 
	When $m< r/3$ we have $\beta(m)=\beta_{ind}(m)=0$. 
	When $r/3\le m \le r$, we have  
	\[
		\frac{(10 h^{3/2}) m}{\sqrt{r/3}}  ~\ge~  (10 h^{3/2}) \sqrt{m}
	\]
	and therefore 
	\[
		\beta(m) ~=~ 0 ~\le~ \frac{(10 h^{3/2}) m}{\sqrt{r/3}}  - (10 h^{3/2}) \sqrt{m} ~=~ \beta_{ind}(m).
	\]
	When $m>r$ we use use the recurrence~\eqref{eq1}, where $m_1,m_2\ge m/3\ge r/3$, and the induction
	hypothesis to obtain
	\begin{align*}
		\beta(m) ~&\le~ \beta (m_1 )+ \beta (m_2) + h^{3/2}\sqrt{m} \\
		&\le~ \frac{(10 h^{3/2}) m_1}{\sqrt{r/3}}  -  (10 h^{3/2}) \sqrt{m_1} + 
			\frac{(10 h^{3/2}) m_2}{\sqrt{r/3}}  -  (10 h^{3/2}) \sqrt{m_2} + h^{3/2}\sqrt{m}\\
		&=~ \frac{(10 h^{3/2}) (m_1+m_2)}{\sqrt{r/3}} - (10 h^{3/2}) (\sqrt{m_1}+ \sqrt{m_2}) + h^{3/2}\sqrt{m}\\
		&\le~ \frac{(10 h^{3/2}) (m+h^{3/2}\sqrt{m})}{\sqrt{r/3}} - (10 h^{3/2}) (\sqrt{m_1}+ \sqrt{m_2}) + h^{3/2}\sqrt{m}\\
		&=~	\beta_{ind}(m) + (10 h^{3/2}) \sqrt{m} + \frac{(10 h^{3/2}) h^{3/2}\sqrt{m}}{\sqrt{r/3}} - (10 h^{3/2}) (\sqrt{m_1}+ \sqrt{m_2}) + h^{3/2}\sqrt{m}\\
		&=~	\beta_{ind}(m) + (11 h^{3/2}) \sqrt{m} + \frac{(10 h^{3/2}) h^{3/2}\sqrt{m}}{\sqrt{r/3}} - (10 h^{3/2}) (\sqrt{m_1}+ \sqrt{m_2}).
	\end{align*}
	To get $\beta(m)\le \beta_{ind}(m)$, it suffices to show that
	\begin{align*}
		(11 h^{3/2}) \sqrt{m} + \frac{(10 h^{3/2}) h^{3/2}\sqrt{m}}{\sqrt{r/3}} ~~\le~~ (10 h^{3/2}) (\sqrt{m_1}+ \sqrt{m_2}).
	\end{align*}
	Dividing by $h^{3/2}\sqrt{m}$, using that $\sqrt{\cdot}$ is concave, and that $m_1+m_2\ge m$
	with $m_1,m_2\ge m/3$, we see that it is enough to show that
	\begin{align*}
		11 + \frac{10 h^{3/2}}{\sqrt{r/3}} ~~\le~~ 10 \left(\sqrt{1/3}+ \sqrt{2/3} \right)~~=~~
		13.9384\dots,
	\end{align*}
	or equivalently
	\begin{equation}
		\frac{h^{3/2}}{\sqrt{r}} ~~\le~~ \frac{10 \left(\sqrt{1/3}+ \sqrt{2/3} \right)- 11}{10 \sqrt{3} } ~=~ 0.1696\dots
		\label{eq:last}
	\end{equation}
	Since $h^{3/2} \le \sqrt{r}/6= \sqrt{r}\cdot 0.1666\dots$, the inequality \eqref{eq:last} holds
	and therefore $\beta(m)\le \beta_{ind}(m)$ for $m\ge r$.
	This finishes the proof by induction that $\beta(m)\le \beta_{ind}(m)$ for all $m$.
\end{appendix}

\end{document}